\documentclass[a4paper]{article}
\usepackage{a4, amsmath, amssymb, graphics, subfigure, fullpage, color}
\usepackage[pdftex]{graphicx}
\usepackage[colorlinks]{hyperref}
\hypersetup{colorlinks,citecolor=green,filecolor=black,linkcolor=blue,urlcolor=blue}
\usepackage{fancyhdr}
\usepackage{multirow}
\usepackage{amsthm}
\newtheorem{theorem}{Theorem}
\newtheorem{definition}{Definition}
\newtheorem{corollary}{Corollary}
\newtheorem{proposition}{Proposition}
\newtheorem{lemma}{Lemma}

\newtheorem{remark}{Remark}

\usepackage{lineno}
\usepackage{upgreek} 
\usepackage{lscape}

\newcommand{\cF}{\mathcal{F}}
\newcommand{\cG}{\mathcal{G}}

\DeclareMathOperator{\E}{\mathbb{E}}
\newcommand{\Q}{\mathbb{Q}}

\newcommand{\beq}{\begin{equation}}
\newcommand{\eeq}{\end{equation}}
\newcommand{\beqn}{\begin{eqnarray}}
\newcommand{\eeqn}{\end{eqnarray}}
\newcommand{\bfig}{\begin{figure}}
\newcommand{\efig}{\end{figure}}
\newcommand{\btab}{\begin{table}}
\newcommand{\etab}{\end{table}}

\ifdefined \up 
	\renewcommand{\up}{\textrm{up}}
\else 
	\DeclareMathOperator{\up}{\textrm{up}} 
\fi

\usepackage{lscape}
\usepackage{lineno}
\usepackage{booktabs}
\usepackage{mathrsfs}
\usepackage{dsfont}
\usepackage{float}
\usepackage{comment}

\newcommand{\InclPart}
\linenumbers
\newcommand{\bed}{\begin{definition}}
\newcommand{\eed}{\end{definition}}

\title{An arbitrage-free conic martingale model with application to credit risk}

\author{Cheikh Mbaye\thanks{Email: \href{mailto:cheikh.mbaye@uclouvain.be}{cheikh.mbaye@uclouvain.be}.} $\qquad$ Fr\'ed\'eric Vrins\thanks{Contact information: Voie du Roman Pays 34, B-1348 Louvain-la-Neuve, Belgium. E-mail: \href{mailto:frederic.vrins@uclouvain.be}{frederic.vrins@uclouvain.be}.}\vspace{0.2cm}\\ Louvain Finance Center (LFIN) \vspace{0.2cm}\\ Universit\'e catholique de Louvain, Belgium}
\date{}
\begin{document}
\maketitle
\begin{abstract}
Conic martingales refer to Brownian martingales evolving between bounds. Among other potential applications, they have been suggested for the sake of modeling conditional survival probabilities under partial information, as usual in reduced-form models. Yet, conic martingale default models have a special feature; in contrast to the class of Cox models, they fail to satisfy the so-called \emph{immersion property}. Hence, it is not clear whether this setup is arbitrage-free or not. In this paper, we study the relevance of conic martingales-driven default models for practical applications in credit risk modeling. 
We first introduce an arbitrage-free conic martingale, namely the $\Phi$-martingale, by showing that it fits in the class of Dynamized Gaussian copula model of Cr\'epey et al., thereby providing an explicit construction scheme for the default time. In particular, the $\Phi$-martingale features interesting properties inherent on its construction easing the practical implementation. Eventually, we apply this model to CVA pricing under wrong-way risk and CDS options, and compare our results with the JCIR++ (a.k.a. SSRJD) and TC-JCIR recently introduced as an alternative.
\end{abstract}
\textbf{Keywords:} default intensity, conic martingale, $\Phi$-martingale, immersion, arbitrage, credit risk.

\ifdefined \InclPart

\section{Introduction}
\label{sec:Intro}


Term-structure models that is, models that allow to generate a set of curves at future times, starting from a given curve at time 0 are particularly popular in interest rates to model discount or forward curves. In this context, the interest rates model is chosen so as to yield a deterministic discount curve at time 0, $P_0(T)$, say, and various discount curves at any time $t>0$, $P_t(T)$, $T\geq t$. Note that the curve $P_0(T)$ is deterministic, but those associated to future times, $P_t(T)$, depend on the evolution of the underlying stochastic model up to $t$. This can be achieved by relying on short-rate models (Vasicek, Hull-White, CIR, etc) or instantaneous forward rates (Heath-Jarrow-Morton (HJM) or market models). Term-structure models are equally useful in credit risk applications, to model the dependency of credit spreads to the maturity. More generally, term-structre models are required to model the default (or survival) probability curve prevailing at time $t$, $Q_t(T)$. Just like interest rates models, we start by assuming a given  curve at time 0 (here, a survival probability  curve $Q_0(T)=G(T)$), and let the stochastic model generate future curves, noted $Q_t(T)$, $T\geq t$.\medskip

Given the well-known equivalence between short-rate interest rates models and intensity-based default models, the machinery developed in the interest rates literature can be recycled in credit risk applications, possibly with some restrictions. Indeed, while negative rates can be tolerated (or even desired), such a thing like a ``negative intensity'' makes no sense. \medskip

In standard reduced-form models, credit risk is handled by modelling the default time as the first jump of a stochastic process. The jump likelihood is controlled  by the intensity process. The later is most of the time stochastic, and in any case is restricted to be positive or, at least, non-negative. This specific setup corresponds to the \emph{Cox framework}. This specific class of reduced-form models is popular due to its tractability and the fact that it satisfies the \emph{immersion property}. The later guarantees the absence of arbitrage opportunities. Following \cite{Jeul80, Elliott00, JeanRut00}, a reduced-form model for defaultable claims can be constructed by considering a full filtration obtained by progressively enlarging the default-free market filtration with the default time. The full filtration is usually considered as the relevant filtration in credit risk models: it represents the information available on the market, to be used for pricing and hedging defaultable claims. When working in such a setup, the most fundamental object attached to the random default time is certainly the conditional survival process, known as the \emph{Az\'ema supermartingale}. Another fundamental behaviour ensuring the no-arbitrage condition in the enlarged filtration is the above-mentioned immersion property which states that the martingales in the default-free filtration remain martingales in the full filtration. For more details about the immersion property and the enlargement of filtration theory we refer the reader to, among others, \cite{Jeanb08} and \cite{jeanblanc2011random}. Alternatively, when considering such a framework, a set of problems concerns the specification of the dynamics of the default intensities. In order to ease the calibration of the model, one naturally choose a specification allowing for an easy solution of the pricing problem. Several routes are possible, but some of them run into other problems. One could think about postulating Gaussian dynamics, but this could mean negative default intensities in a large number of classes. Yet, from a practical perspective, it is common to consider the homogeneous affine term structure models with positive dynamics such as CIR (see \cite{Duffie99}) or JCIR (Cox-Ingersoll-Ross with independent compound Poisson jumps, see \cite{Brigo10}) to deal with the intensity process. Although very popular, these models present some drawbacks. The classical time-homogeneous affine models such as CIR do not have enough flexibility when it comes to perfectly fit a given market curve. To circumvent this issue, these models were extended by starting with a non-negative time-homogeneous affine (i.e. very tractable) model and adding a deterministic shift, leading to the well known CIR++ (SSRD) or JCIR++ (SSRJD) introduced in \cite{Brigo06} and further studied in \cite{Brigo05} and \cite{Brigo10} for specific applications in credit derivatives. The shift approach is appealing since it solves the perfect fit problem and preserves the affine property of the dynamics. However, when shifting the intensity process in a deterministic way, there is no guarantee that the resulting process remains positive when forcing the fit to a given market curve. This problem can be handled by including a non-negativity constraint on the shift function when optimizing the parameters of the time-homogeneous intensity process. More recently, the calibration problem has been solved using a different deterministic adjustment: the shift function is replaced by a time change. These two modifications of the shift extension are refered to as the positive-shift (PS-(J)CIR) and the time-changed  (TC-(J)CIR) (J)CIR  \cite{Mbaye2019}. We refer to \cite{Duffie96,Duffie03} for general results within the affine term structure models and to \cite{Mbaye2019} regarding the perfect fit problem.\medskip

Interestingly, all these models fit in the class of Cox models. Indeed, in all these cases, the associated Az\'ema supermartingale is decreasing; its Doob-Meyer decomposition exhibits no  martingale component. This shows that such models actually correspond to a very special case. Yet, the above property is interesting: a vanishing martingale part in the Doob-Meyer decomposition of the Az\'ema supermartingale proves the associated models to be arbitrage-free. Although a few models featuring a martingale component in their Az\'ema supermartingale have been discussed in the literature \cite{Biel08, Vrins16, CrepSong14}, little work has been done to actually make ``non-Cox models'' workable.\medskip

In this paper, we deal with a class of default models, \emph{conic martingales} or the \emph{martingale approach}, recently introduced in \cite{Vrins14} and \cite{Vrins15a} and further developed in \cite{Vrins16}. Conic martingales offer a modelling framework that completely gets out of Cox models. It consists of a direct modeling of the Az\'ema supermartingale and is a setup where immersion property does not hold. As explained above, it is therefore not clear whether this model is free of arbitrage opportunities. 
A central point in this paper is indeed to give an answer to this question but in a more practical perspective. While very promising, conic martingales trigger important mathematical challenges and deserve in depth technical analysis when dealing with arbitrages opportunities in a no-immersion setup. To fill this gap, we rely on recent results introduced by Cr\'epey et al. \cite{CrepSong14} to show that a special case of conic martingales, the $\Phi$-martingale, belongs to another class of default models that are arbitrage free. In particular, we pay attention to the fact that the $\Phi$-martingale possesses interesting analytical properties which rends it quite suitable for credit risk applications.\medskip

In Section \ref{sec:gdm}, we recall some standard Cox models (like JCIR++, TC-JCIR and HJM) and briefly review the martingale approach with a particular focus on the $\Phi$-martingale case. For both models, we provide the default time definition. Section \ref{sec:arbitrage} introduces the study of the no-arbitrage property under immersion and beyond immersion. The arbitrage free property of the $\Phi$-martingale model is then established. In Section \ref{sec:numerics}, we compare the performances of the $\Phi$-martingale model with the JCIR++ and TC-JCIR models to the pricing of two classes of credit derivatives: credit valuation (CVA) adjustment under the presence of wrong-way risk and credit default swap (CDS) option, before concluding in section \ref{sec:concl}.

\section{Intensity-based default models}\label{sec:gdm}
Throughout the paper, we consider a fixed time horizon $T^*$ and a probability space $(\Omega, \mathcal{G},\mathbb{Q})$. Our financial market can feature default-free entities and, to ease the exposition, a single credit-risky entity, which default time is modeled by the random time $\tau$. Hence, we deal with two classes of financial instruments: those which future cashflows (hence prices) are not impacted by the default of the risky reference entity (called \textit{default-free assets} in the sequel), and those who are (\textit{defaultable assets}). To deal with those products, we consider several flows of information, modeled as filtrations satisfying the usual conditions. They will be formally specified for each model below, but they can be intuitively introduced as follows. The full market information is noted $\mathbb{G}=(\mathcal{G}_t, t\in [0,T^*])$. In this paper, all risk factors and price processes are $\mathbb{G}$-adapted. Then, we define the filtration specific to the default event, i.e. the natural filtration of the default indicator $\mathbb{D}=(\mathcal{D}_t, t\in [0,T^*])$ ,  $\mathcal{D}_t=\sigma(\mathds{1}_{\{\tau\leq u\}}, u\leq t, t\in [0,T^*])$. Eventually, $\mathbb{F}=(\mathcal{F}_t, t\in [0,T^*])$ is a sub-filtration of $\mathbb{G}$. Loosely speaking, it is defined as the largest sub-filtration of $\mathbb{G}$ such that $\tau$ is a $\mathbb{F}$- but not a $\mathbb{G}$-stopping time. Notice that $\mathbb{F}$ should not be considered as the information conveyed by the risk factors driving the default-free assets only. Indeed, some processes impacting the default likelihood could be $\mathbb{F}$-adapted, too. The important thing is that, given $\mathcal{F}_t$, it should not be possible to determine whether the default event took already place or not.
We assume in the sequel that $\mathcal{G}=\mathcal{G}_{T^*}$ and  $\tau>0$. The probability $\mathbb{Q}$ stands for an equivalent martingale probability measure, so that every payoff discounted at the $\mathbb{F}$-adapted risk-free rate $r$ is a $(\mathbb{Q},\mathbb{G})$-martingale. Notice that the discounted prices of default-free assets (which future cashflows do not depend on $\tau$) are $(\mathbb{Q},\mathbb{F})$-martingales as well. Eventually, we assume that the market provides us with the (risk-neutral) survival curve $G$, which represents the current $\mathbb{Q}$-distribution that the reference entity survives up to  some point in time, i.e.
 $$G(T)=\Q(\tau>T|\cG_0)=\Q(\tau>T)\;.$$

 In practice, the $G$ curve is obtained by a bootstrapping procedure, that is, by considering the market prices of defaultable instruments like credit-risky bonds or credit default swaps (CDS), and reverse-engineering the risk-neutral valuation formula iteratively, for increasing maturities. Obviously, the $G$ function must start from 1, remain positive and be decreasing. As standard in the literature and in line with the market practice, we assume that $G$ is differentiable~\cite{IsdaCDS}. Therefore, the market-implied survival probability curve observed at time $0$ can be parametrized as
\begin{equation}
\label{eq:calib}
G(t)=e^{-\int_0^th(s)ds}\;,
\end{equation} for some positive function $h$ called \textit{hazard rate}.\medskip

The purpose of a dynamic default model is to generate a set of probability curves at some future time $t>0$, that is, to model the probability that a default event occurs after a given time $T\in[t,T^*]$ given the information available at time $t\leq T$. Mathematically speaking, the model aims at providing
\begin{equation}
Q_t(T) := \mathbb{Q}(\tau> T|\mathcal{G}_t) 
\;.
\end{equation}

These curves are needed for pricing (e.g. options on CDS or credit valuation adjustment) or risk-management purposes. We refer to \cite{Ces09} for a couple of examples.

\subsection{Conditional survival probabilities under partial information}

Although prices are given by considering the full market information that is, by computing $\mathbb{G}$-conditional expectations, the considered setup allows us to work in the sub-filtration $\mathbb{F}$ thanks to the \emph{Key lemma}. More explicitly, the above $\mathcal{G}_t$-conditional probability can be written as a ratio of $\mathcal{F}_t$-probabilities, scaled by a survival indicator \cite[Lemma 3.2.1.]{Biel11}: 
\begin{equation}
\mathbb{Q}(\tau> T|\mathcal{G}_t)
=\mathds{1}_{\{\tau>t\}}\frac{\mathbb{Q}(\tau> T|\mathcal{F}_t)}{\mathbb{Q}(\tau> t|\mathcal{F}_t)}\;.
\end{equation}
Introducing the following notation for the $\mathcal{F}_t$-conditional survival probability curve
\begin{equation}
\label{eq:StT}
S_t(T) :=\mathbb{Q}(\tau> T|\mathcal{F}_t) = \mathbb{E}\left[\mathds{1}_{\{\tau>T\}}|\mathcal{F}_t\right]
\;,
\end{equation}
one gets that the $\mathcal{G}_t$-risk-neutral probability of the event $\{\tau>T\}$, $T\geq t$, can be written as 
\begin{equation}
\label{eq:condsurv}
Q_t(T) :=  \mathds{1}_{\{\tau>t\}}\frac{S_t(T)}{S_t(t)}\;.
\end{equation}
Interestingly, for every $T\in[0,T^*]$, $\left(S_t(T)\;,~t\in[0,T]\right)$ is a   $(\mathbb{Q},\mathbb{F})$-martingale valued in $[0,1]$. By contrast, 
\begin{equation}
\label{eq:Stt}
S_t:=S_t(t)= \mathbb{E}\left[\mathds{1}_{\{\tau>t\}}|\mathcal{F}_t\right]
\end{equation}
is also valued in $[0,1]$, but is a $(\mathbb{Q},\mathbb{F})$-supermartingale. Indeed, from the tower law, we have for $s\geq t$,

$$\mathbb{E}[S_s|\mathcal{F}_t]=\mathbb{E}\left[[\mathds{1}_{\{\tau>s\}}|\mathcal{F}_s]|\mathcal{F}_t\right]=\mathbb{E}[\mathds{1}_{\{\tau>s\}}|\mathcal{F}_t]=\mathbb{Q}(\tau>s|\mathcal{F}_t)\leq \Q(\tau>t|\mathcal{F}_t)=S_t\;,$$
since $\{\omega\in\Omega:\tau(\omega)>s\}\subseteq\{\omega\in\Omega:\tau(\omega)>t\}$.
The $S=(S_t,t\in [0,T^*]$ process is often referred to as \textit{survival process}, but is also known as the \emph{Az\'ema supermartingale} in the probability literature. Clearly, $S_0=S_0(0)=1$ from  \eqref{eq:Stt} because $\tau>0$ and, from \eqref{eq:condsurv}, $Q_0(T)=S_0(T)=G(T)$ where the last equation comes from the calibration procedure at time $0$. Notice that from the tower law again, the expectation of the survival process at time $T$ is nothing but the probability that the reference entity survives up to $T$, as seen from time $t=0$:
$$\mathbb{E}[S_T]=\mathbb{E}\left[\mathbb{E}[\mathds{1}_{\{\tau>T\}}|\cF_T]\right]=\mathbb{E}\left[\mathds{1}_{\{\tau>T\}}\right]=\mathbb{Q}(\tau>T)=G(T)\;.$$
One can consider this expression as a kind of constraint (or calibration) that the default model must satisfy at time 0. Note that not all models generate a curve $\mathbb{E}[S_\cdot]$ compatible with the form of $G(\cdot)$ given in \eqref{eq:calib}. It is however the case for all models such that $\tau$ is a continuous random variable satisfying $\Q(\tau\leq T^*)<1$. In this case, $\tau$ admits a density, $\alpha$ and
$\E[S_t]=\Q(\tau>t)=1-\int_0^t \alpha(s) ds>0$ for all $t\in[0,T^*]$. This expectation can be written as in \eqref{eq:calib} provided that $h(t):=\frac{\alpha(t)}{1-\int_0^t \alpha(s) ds}$. This will be the case in all models considered below.\medskip

A fundamental result from stochastic calculus stipulates that any survival process $S$ admits a unique Doob-Meyer decomposition~\cite{Dell75},~\cite{Biel11} 
\begin{equation}
\label{eq:doob}
S_t = A_t+M_t\;,
\end{equation}
where $M$ is a $(\mathbb{Q},\mathbb{F})$-martingale and $A$ is an $\mathbb{F}$-predictable decreasing process, satisfying $M_0=0$ and $A_0=1$. If in addition $S$ is continuous, then so are $A$ and $M$ \cite{Biel08}. Moreover, if $A$ is absolutely continuous with respect to the Lebesgue measure, $dA_t=-\mu_tdt$ and $dM_t=\sigma_tdB_t$ where $\mu$ is a positive process, $\mu,\sigma$ are $\mathbb{F}$-adapted and $B$ a $(\mathbb{Q},\mathbb{F})$-Brownian motion \cite{Biel08}. Hence, whenever $S_t>0$ for every $t\in]0,T^*]$, then the dynamics of the survival process can be written as
\begin{equation}
\label{eq:azema}
dS_t=-\lambda_tS_tdt+\sigma_tdB_t, \quad S_0=1,
\end{equation}
$\lambda_t:=\mu_t/S_t$ is an $\mathbb{F}$-progressively measurable non-negative process called \emph{default intensity}.

%

\begin{remark}
It is common to expect the survival process $S$ to be decreasing. This is a feature that is indeed met in the \emph{usual} default models. Surprisingly or not, this is just a special case: it is clear from the calibration procedure that the expectation of $S$ is decreasing but, from \eqref{eq:azema}, the process $S$ is decreasing if and only if $M\equiv 0$, i.e. $\sigma\equiv 0$. 
At this stage, observe that $S_t(T)$ in \eqref{eq:StT} is decreasing with respect to $T$ for $T\in[t,T^*]$ but is a $(\Q,\mathbb{F})$-martingale. In particular, it does \emph{not} decrease with $t$. This behavior simply results from the fact that we are computing probabilities under partial information and, in such circumstances, one is allowed to change her mind about past events, at least as long as those events remain \emph{unobserved}, i.e. as long as they are not measurable.
\end{remark}


We now recall three different models and introduce a new one. The first model is the well-known deterministic shift extension to the Cox-Ingersoll-Ross model with compound Poisson jumps (JCIR++ or SSRJD) extensively studied in \cite{Brigo06}. It will serve as comparison for the other \---more recent\--- models further considered. The second model is a time-changed version of the JCIR, called TC-JCIR, introduced in \cite{Mbaye2019}. It is indeed an interesting alternative to the JCIR++ model. The third model is a defaultable Heath-Jarrow-Morton (HJM) model~\cite{Biel11}. As shown below, all these approaches are \emph{Cox models}. Eventually, the last model, which is of interest here, is based on conic martingales originally introduced in \cite{Vrins14} and further studied in \cite{Vrins16}. For each model, we derive the $\mathbb{F}-$ and $\mathbb{G}$-conditional survival probability curves (i.e. $S_\cdot(T)$ and $Q_\cdot(T)$), as well as the Az\'ema supermartingale ($S_\cdot$).\medskip

\subsection{Reduced-form approach: the Cox setup}

The \textit{Cox setup} is the most popular approach for dynamic intensity-based (reduced form) modelling. It is originally due to \cite{Lando98} and \cite{Duffie99}. We refer to \cite{Brigo06} for extensive applications in credit risk. In contrast with the firm-value (from which default occurs when the firm's asset breaches a default barrier, a.k.a. Black-Cox or structural models \cite{Merton74}) the default event is triggered by the first jump of a counting process with stochastic intensity $\lambda$.
Equivalently, $\tau$ can be modelled as the first passage of $\int_0^\cdot \lambda_sds$
above a random threshold $\mathcal{E}$:
\begin{equation}
\label{eq:defTauCox}
\tau  := \inf\left\lbrace t\geq 0 : \Lambda_t\geq \mathcal{E}\right\rbrace,~~\Lambda_t:=\int_0^t \lambda_s ds\;.
\end{equation}

In this model, $\lambda$ is a non-negative, $\mathbb{F}$-adapted process and $\mathcal{E}$ is a random variable with unit exponential distribution, independent from $\mathcal{F}_{T^*}$. Hence, one can choose as $\mathbb{F}$ the natural filtration of $\lambda$ (possibly enlarged with the factors impacting the default-free assets), $\mathbb{D}$ collapses to the filtration generated by the pair $(\lambda,\mathcal{E})$ and $\mathbb{G}=\mathbb{F}\vee\mathbb{D}$.\medskip

Because $\lambda$ is positive $\Q$-a.s., $\Lambda$ is increasing, hence, the survival process $S$ defined in \eqref{eq:Stt} reduces to
\begin{equation}
S_t=\mathbb{Q}(\tau>t|\mathcal{F}_t)=\Q(\Lambda_t\leq \mathcal{E}|\mathcal{F}_t)=e^{-\Lambda_t}=e^{-\int_0^t\lambda_sds}\;.
\end{equation}
The dynamics of the survival process are given by
\begin{equation}
\label{eq:azemaCox}
dS_t=-\lambda_tS_tdt\;,
\end{equation} showing that  $\lambda$ in \eqref{eq:defTauCox} actually corresponds to the default intensity introduced in \eqref{eq:azema}. Moreover, the Cox setup corresponds to a very special Doob-Meyer decomposition: it deals with decreasing survival processes, i.e. with $S$ having no martingale part ($M\equiv 0$).\medskip

It is also easy to compute the conditional survival probabilities under both filtrations. For instance, the $\mathcal{F}_t$-conditional survival probability that $\tau>T$ is obviously a $(\mathbb{Q},\mathbb{F})$-martingale on $[0,T]$, and reads as
\begin{equation}
S_t(T)=\mathbb{Q}(\tau>T|\mathcal{F}_t)=\Q(\Lambda_T\leq \mathcal{E}|\mathcal{F}_t)=e^{-\Lambda_t}\mathbb{E}\left[\left.e^{-\int_t^T\lambda_s ds}\right|\mathcal{F}_t\right]\;.\label{eq:StTCox}
\end{equation}
The corresponding $\mathcal{G}_t$-conditional survival probability is, from \eqref{eq:condsurv}, given by 
\begin{equation}
Q_t(T)=\mathbb{Q}(\tau>T|\mathcal{G}_t) = \mathds{1}_{\{\tau>t\}}e^{\Lambda_t}\mathbb{E}\left[\left.e^{-\Lambda_T}\right|\mathcal{F}_t\right]= \mathds{1}_{\{\tau>t\}}\mathbb{E}\left[\left.e^{-\int_t^T\lambda_sds}\right|\mathcal{F}_t\right]\;.\label{eq:QtTCox}
\end{equation}
In the special case where $\lambda$ is an affine process, the above conditional expression takes the usual exponential-affine form:
\begin{equation}
S_t(T) =e^{-\Lambda_t}P^\lambda_t(T,\lambda_t)\;,~~ Q_t(T) =\mathds{1}_{\{\tau>t\}}P^\lambda_t(T,\lambda_t)
\end{equation}
where 
$$P^x_t(T,z):=
A^x(t,T)e^{-B^x(t,T)z}$$
for $0\leq t\leq T\leq T^*$ and some deterministic functions $A^x$ and $B^x$ (we refer to \cite{Brigo06} for more details). To ease the notation, we set $P^x(t):=P^x_0(t,x_0)$.\medskip

We give below to examples of reduced-form models.

\subsubsection{JCIR++}
\label{sec:jcir}
The JCIR++ model 
postulates the following dynamics for the intensity process:
\begin{equation}
\lambda^\varphi_t = x_t + \varphi(t)
\end{equation}
where $\varphi$ is a deterministic function and $x$ is a time-homogeneous JCIR model
\begin{equation}
\label{eq:jcir}
dx_t=\kappa(\beta - x_t)dt + \delta\sqrt{x_t}dB_t + dJ_t, \quad x_0\geq 0
\end{equation}
with $\kappa$, $\beta$, $\delta$ some positive constants and $J$ is a compound Poisson process with jump intensity $\omega\geq 0$ and exponential jump size with mean $1/\alpha$, $\alpha>0$, independent of $B$.\medskip

In this model, $\tau$ is defined as in \eqref{eq:defTauCox}  but with intensity $\lambda\leftarrow\lambda^\varphi$. Therefore, $\mathbb{F}$ is chosen to be the natural filtration of $x$ (possibly enlarged with the factors impacting the default-free assets), while $\mathbb{D}$ and $\mathbb{G}$ are as before.\medskip 

The shift function $\varphi$ is used in the calibration step. Its purpose is to guarantee a perfect fit between the survival probability $\mathbb{Q}(\tau>t)$ implied by the model with the curve $G(t)$ extracted from market data.
Mathematically, it is determined such that $\mathbb{E}[S_t]=G(t)$ for all $t\in[0,T^*]$. This identifies the prevailing shift function for a given set $(x_0,\kappa,\beta,\delta)$, which takes a well-known expression:
\begin{equation}
\mathbb{E}[S_t]=P^{\lambda^\varphi}(t)=e^{-\int_0^t\varphi(s)ds}P^x(t) = G(t)~~\Rightarrow~~\varphi(t)=-\frac{d}{dt}\ln\frac{G(t)}{P^x(t)}\;.\label{eq:CalJCIR}
\end{equation}

The conditional survival probabilities that $\tau>t$ become 

$$S_t(T)=e^{-\int_0^T\varphi(s)ds}e^{-\int_0^tx_sds}P^x_t(T,x_t)=\frac{G(T)}{e^{\int_0^tx_sds}}\frac{P^x_t(T,x_t)}{P^x(T)}\;,$$
and 
$$Q_t(T)=\mathds{1}_{\{\tau>t\}}\frac{G(T)}{G(t)}\frac{P^x(t)}{P^x(T)}P^x_t(T,x_t)\;.$$

This model has the advantage of 
being able to perfectly fit any (continuous) survival probability curve $G$ implied from the market without affecting analytical tractability (in terms of prices of zero-coupon  bonds and European options). Moreover, thanks to the jump process $J$, it can generate large implied volatilities without breaking Feller's constraint, i.e. such that the origin is not accessible for $\lambda$. Unfortunately, it suffers from an important drawback: we cannot guarantee the positiveness of the intensity process $\lambda^\varphi$ (hence, of $\lambda$ since $x$ can be arbitrarilly close to 0) without ad-hoc constraints when computing the parameters $(x_0,\kappa,\beta,\delta)$. This drawback becomes more and more serious when increasing the activity of the jump process $J$ since only positive jumps are allowed for tractability reasons. Therefore, increasing the jump activity under the constraint that $\mathbb{E}[S_t]=G(t)$ for a given curve $G$ requires to lower the shift, possibly to the negative territory. 
We refer to \cite{Mbaye2019} for a detailed analysis of the negativity intensity issue of the JCIR++ model.   

\subsubsection{TC-JCIR}
\label{sec:tcjcir}
The TC-JCIR model is an alternative to the JCIR++ aiming to solve the negative intensity issue without losing neither the analytical tractability nor the calibration flexibility of the JCIR++. The model flexibility is achieved by time-changing the non-negative $x$-model in a deterministic way. Therefore, although similar in principle with the deterministic shift extension of time-homogeneous models introduced above, the positiveness of the intensity process is guaranteed by construction. More specifically, the intensity is modeled as 

$$\lambda^\theta_t=\theta(t)x^\theta_t\;,~~x^\theta_t:=x_{\Theta(t)}\;,$$
where $x$ is a time-homogeneous non-negative affine model (e.g. JCIR), $\Theta(t)$ is a time change function called a \emph{clock} and  $\theta(t):=\Theta'(t)$ is the \emph{clock rate}. In this model, $\tau$ is defined as in \eqref{eq:defTauCox}  but with $\lambda\leftarrow\lambda^\theta$. Therefore, $\mathbb{F}^\theta:=(\mathcal{F}_{\Theta(t)})_{t\in[0,T]}$ is chosen to be the natural filtration of $x^\theta$ (possibly enlarged with the factors impacting the default-free assets), while $\mathbb{D}$ and $\mathbb{G}$ are given by their corresponding time change filtrations.\medskip

The clock plays a similar role as the shift in the JCIR++ model: it is chosen such that
\begin{equation}
\mathbb{E}[S_t]=P^{\lambda^\theta}(t)=P^x(\Theta(t)) = G(t)~~\Rightarrow~~\Theta(t)=Q^x(G(t))\;,\label{eq:CalTCCIR}
\end{equation}
where $Q^x$ is the inverse of $P^x$. 
We refer to \cite{Mbaye2019} for more details about this model.\medskip

The conditional survival probabilities that $\tau>t$ are given by 


$$S_t(T)=\mathbb{E}\left[\left.e^{-\Lambda^\theta_T}\right|\mathcal{F}_t\right]=\mathbb{E}\left[\left.e^{-\int_0^{\Theta(T)}x_sds}\right|\mathcal{F}_t\right]=\frac{P^x_{Q^x(G(t))}(Q^x(G(T)),x_{Q^x(G(t))})}{\exp\lbrace\int_0^{Q^x(G(t))}x_sds\rbrace}\;,$$

and

$$Q_t(T)=\mathds{1}_{\{\tau>t\}}P^x_{Q^x(G(t))}(Q^x(G(T)),x_{Q^x(G(t))})\;.$$

It can be shown that the clock solving equation \eqref{eq:CalTCCIR} takes the form $\Theta(t):=\int_0^t \theta(s) ds$ where $\theta$ is non-negative, leading to a valid time change function. Hence, $\mathbb{Q}(\lambda_t\geq 0)=1$ for all $t$, solving the negative intensity issue. Interestingly, the process $x^{\theta}$ remains affine if so is $x$, although not necessarily time-homogeneous affine, obviously. In the sequel we take consider JCIR dynamics for $x$, i.e. the same as for the JCIR++, for the sake of comparison.



\subsubsection{HJM intensities}


The general expression of $Q_t(T)$ in \eqref{eq:QtTCox} suggests that the reduced-form models (and JCIR++ and TC-JCIR in particular) can be considered as \emph{short intensity} models, by analogy with \emph{short rate} models in the interest rates literature. Indeed, the conditional expectation agrees with the time-$t$ no-arbitrage price of a default-free zero-coupon bond price with maturity $T$ provided that $\lambda$ stands for the short risk-free rate. It is possible to revisit these models \emph{\`{a} la Heath-Jarrow-Merton}, by modeling directly the term structure of the future default intensities, i.e. by modeling the hazard rate \emph{curve} at once.\medskip

In \cite{Schon88}, Schonbucher models the default-free and defaultable instantaneous forward rate curves with a same Brownian motion. 
The instantaneous forward curve associated with the risk-free rate is noted $f_t(u)$. The time-$t$ no-arbitrage price of a default-free zero-coupon bond price with maturity $T$ becomes $e^{-\int_t^T f_t(u)du}$ as the function $f_t$ is $\mathcal{F}_t$-measurable. This expression can take a similar form to the short-rate expression
\begin{equation}
    \label{eq:PtT}
    P_t(T)=\E\left[\left.e^{-\int_t^T r_s ds}\right|\mathcal{F}_t\right]
\end{equation}
 provided that we set $r_t:=f_t(t)$ with initial condition $f_0(T)=-\left.\frac{d}{du}P_0(u)\right|_{u=T}$. In this setup, only the diffusion coefficients of $f_\cdot(T)$ need to be specified; the drift is given by a no-arbitrage argument. A similar term structure model is assumed for the instantaneous forward curves associated with the \emph{defaultable} instruments, $\bar{f}_t(T)$. It turns out that defining the \emph{credit spread} process $\lambda_t(T):=\bar{f}_t(T)-f_t(T)$ ($0\leq t\leq T\;,~~T\in[0,T^*]$), the process $\lambda_t:=\lambda_t(t)$ is strictly positive, $\mathbb{Q}$-a.s. In fact, the latter can be interpreted as the default intensity, in the sense that the default time can be defined as in \eqref{eq:defTauCox}.\medskip 

A slightly different point of view is considered in \cite{Chiar11} where the author starts from a similar setup but model $f_t(T)$ and $\lambda_t(T)$ with two correlated Brownian motion to obtain the dynamics of $\bar{f}_t(T)$ satisfying $\bar{f}_t:=\bar{f}_t(t)=f_t(t)+\lambda_t(t)=r_t+f_t$. Here again, $\lambda_t(T)$ can be interpreted as the $\mathcal{F}_t$-measurable hazard rate curve prevailing at time $t$. In either HJM frameworks, the drift of $\bar{f}_\cdot(T)$ (hence that of $\lambda_\cdot(T)$) is given by no-arbitrage, and it holds that 

$$Q_t(T)=\mathds{1}_{\{\tau>t\}}\mathbb{E}\left[\left.e^{-\int_t^T \lambda_sds}\right|\mathcal{F}_t\right]=\mathds{1}_{\{\tau>t\}}e^{-\int_t^T \lambda_t(s)ds}\;.$$


Compared to the JCIR++ model, HJM models are appealing for several reasons. First, the calibration equation $\E[S_t]=G(t)$ is automatically satisfied by imposing the initial condition $\lambda_0(T)=h(T)$: $$\Q(\tau>T)=\Q(\tau>T|\mathcal{G}_0)=Q_0(T)=e^{-\int_0^T h(s)ds}=G(T)\;.$$
The problem of negative intensities in the JCIR++ model can thus be ruled out by choosing appropriate dynamics for $\lambda_\cdot(T)$, as the calibration equation is handled by the initial condition. Second, one can directly model the shape of the volatility of instantaneous forward intensities via the diffusion coefficients. It is worth pointing out that although this can be appealing for the sake of dealing with risky rates (like Libor rates, as in \cite{Fanel16}), it is probably less relevant for pure credit applications due to the scarcity of quotes on the credit options' market. Moreover, as pointed out in \cite{Fanel16}, this model  is not easy to deal with in practice. Anyway, as shown above, it still fit in the same class of Cox models (just like JCIR++ and TC-JCIR). Instead, we consider a similar \--- but different \--- alternative, called \emph{martingale approach}.   

\subsection{Conic martingale approach}
\label{sec:cm}

In contrast with the models introduced above, the martingale approach is a framework that does not fit in the Cox setup. Instead of defining $\tau$ directly as in \eqref{eq:defTauCox} for Cox models or via an intensity, this approach consists of modeling the  $\mathbb{F}$-conditional survival probability curves \eqref{eq:StT} directly using diffusion martingales,

\begin{equation}
dS_t(T)=\sigma(t,S_t(T))dB_t\;,~~0\leq t\leq T \leq T^*\;.\label{eq:dStT}    
\end{equation}
The requirement $\mathbb{E}[S_T]=G(T)$ for all $T\in[0,T^*]$ is automatically satisfied by choosing the initial condition $S_0(T)=G(T)$. 
\begin{remark}Because we are modeling the $\mathbb{F}$-conditional probabilities with the help of the Brownian motion $B$, the latter must be $\mathbb{F}$-adapted. Hence, $\mathbb{F}$ can be chosen as the natural filtration of $B$ (possibly enlarged with the factors impacting the default-free assets). Notice that we do not provide an explicit construction scheme for $\tau$. Therefore, at this stage, we cannot specify explicitly the filtrations $\mathbb{D}$ and $\mathbb{G}$. This is a major drawback of the martingale models. This point will be addressed later in the paper in the particular case of $\Phi$-martingales. 
\end{remark}

To the best of our knowledge, the martingale approach was first considered in \cite{Ces09} in the context of counterparty risk on credit derivatives. They postulate a diffusion coefficient of the form $\sigma(t,T)$ \cite[section 3.3.1]{Ces09}. Obviously, because of the lack of state-dependency, this setup leads to Gaussian dynamics, $$S_t(T)=S_0(T)+\int_0^t\sigma(s,T)dB_s
~\sim\mathcal{N}\left(S_0(T),\int_0^t\sigma^2(s,T)ds\right)\;.$$ 
Just like the JCIR++, TC-JCIR and HJM models introduced above, the model can be perfectly calibrated to the market: imposing the initial condition $S_0(T)=G(T)$ leads to as $\E[S_T]=G(T)$. 
Notice that \eqref{eq:dStT} is a \emph{family} of SDEs. For each $T\in[0,T^*]$, the process $S_t(T)$, $0\leq t\leq T$ represents the evolution of the $\mathcal{F}_t$-conditional probability that $\tau>T$. It is obviously not a binary process since $\tau$ is \emph{not} an $\mathbb{F}$-stopping time.\footnote{Observe from \eqref{eq:StTCox} that this feature is not a specificity of the conic martingale approach. It results from the definition of $S_t(T)$, and is obviously shared by Cox models.\medskip} Nevertheless, each of those processes must belong to the  interval $[0,1]$. This is clearly violated by Gaussian dynamics. Specifying the shape of the diffusion coefficient such that $S_\cdot(T)\in[0,1]$ and, at the same time, get a tractable model is not easy. To circumvent this problem, we can directly model $S_\cdot(T)$ with conic martingales, i.e., with  martingales evolving within a specific range. 
Credit risk has been mentioned as a potential application for such processes, but without being further developed \cite{Vrins16}.\medskip

Following \cite{Vrins16}, the $\mathbb{F}$-conditional probability curves are modeled in one go. To make sure that $S_t(T)\in[0,1]$, we start from a family of latent processes, and map them into a function with appropriate image:
 \begin{equation}
\label{eq:cm}
S_{t}(T) := F(Z_{t,T})\;,
\end{equation}
where $F:\mathbb{R}\longrightarrow [0,1]$ is a $\mathcal{C}^2$ invertible function and $Z_{t,T}$, $0\leq t\leq T, T\leq T^*$, a family of diffusions driven by the same Brownian motion $B$,\footnote{The more general case where the diffusion coefficient reads $\eta(t,T,z)$ can also be dealt with, but is more involved and is not further developed here.}
\begin{equation}
\label{eq:Ztu}
dZ_{t,T} = a(t,Z_{t,T})dt + \eta(t,Z_{t,T})dB_t\;,~~Z_{0,T}:=F^{-1}(G(T))\;.
\end{equation}

This is similar in spirit to one-factor HJM models in the sense that we directly model probability curves with the help of a single Brownian motion. Moreover, the drift in \eqref{eq:Ztu} is uniquely determined by the martingale property of $S_{t}(T)$:
\begin{equation}
\label{eq:drift}
a(t,z) = \frac{\eta^2(t,z)}{2}\psi(z)\;,~~\psi(z):=-F''(z)/F'(z)\;.
\end{equation}
This is a simple consequence of It\^{o}'s lemma (see \cite{Vrins16} for more details). \medskip

However, there is a fundamental difference with HJM models: the martingale approach does not belong to the class of Cox models. Indeed, in contrast with intensity models where $S$ is decreasing \eqref{eq:azemaCox}, the decomposition of $S$ in the the martingale approach \emph{does} feature a non-zero martingale part. In the conic martingale case for instance,\footnote{Notice that it is not enough to replace $T$ by $t$ to get the dynamics of $Z_{t,t}$, as it corresponds to the dynamics of $Z_{t,T}$ for a fixed $T$. The dynamics of $S$ are obtained by applying It\^{o}'s lemma to $F(Z_{t,t})$ with $F(Z_{0,t})=G(t)$.}

$$dS_t=-\frac{F'\left(F^{-1}(S_t)\right)}{F'\left(F^{-1}(G(t))\right)}h(t)G(t)dt+F'\left(F^{-1}(S_t)\right)\eta\left(t,F^{-1}(S_t)\right)dB_t\;.$$

In the sequel, we assume that the diffusion coefficient of $Z_{t,T}$ is a bounded function of time, i.e. $\eta(t,z)=\eta(t)$ where $0<\eta^2(t)<\infty$ on $[0,T^*]$.  Moreover, we assume that the score function $\psi$ of the mapping $F$ is Lipschitz continuous. Then, for each $u\leq T^*$, $0\leq t\leq u$, the SDE
 \begin{equation}
dZ_{t,u} = \frac{\eta^2(t)}{2}\psi(Z_{t,u})dt + \eta(t)dB_t\label{eq:diffZ}
\end{equation}
has a strong, pathwise unique solution on $[0,T^*]$ (see Kloeden-Platten \cite{Kloed99}). 
Since $F$ is a bijection, it is invertible, and the diffusion coefficient of $S_t(T)$ in \eqref{eq:dStT} takes the form $\sigma(t,z)=\eta(t) F'\left(F^{-1}(z)\right)$.\medskip

\subsubsection{The $\Phi$-martingale default model}

A special case consists of considering $F=\Phi$, the cumulative distribution of the standard normal random variable. This is a very particular case where $Z_{\cdot,T}$ are Gaussian processes. Indeed, the  score function $\psi$ collapses to the identity. Therefore, the drift in \eqref{eq:Ztu} is linear and the diffusion coefficient is a bounded, implying that the SDE admits a unique strong solution, and leading to a tractable model. In this model, the process $S_\cdot(T)$ corresponds to $\Phi$-martingale \cite{Vrins16}.

\begin{definition}[The $\Phi$-martingale default model]
\label{def:phim}
The $\Phi$-martingale model postulates that $S_\cdot(T)$ is a Gaussian process mapped to the standard normal cumulative distribution function, i.e.
\begin{equation}
\label{eq:phimart}
S_{t}(T) = \Phi(Z_{t,T})
\end{equation}
with
\begin{eqnarray}
Z_{t,T}&=&\Phi^{-1}(G(T))e^{\int_0^t\frac{\eta^2(s)}{2}ds} + \int_0^t\eta(s)e^{\int_s^t\frac{\eta^2(u)}{2}du}dB_s\label{eq:SolZtT}\\
&\sim&\mathcal{N}\left(\Phi^{-1}(G(T))e^{\int_0^t\frac{\eta^2(s)}{2}ds},~e^{\int_0^t\eta^2(s)ds}-1\right)
\;.\label{eq:SolZtTPhi} 
\end{eqnarray}
\end{definition}

This model is of particular interest for several reasons. First, when $\eta\equiv1$, the process \eqref{eq:phimart} can be seen as the analog of the Brownian motion (martingale valued on $\mathbb{R}$) or its Dol\'eans-Dade exponential (martingale valued in $\mathbb{R}^+$) but for the $[0,1]$ range; see \cite{Vrins16} for a discussion. This feature has been noticed independently by Carr, who called the corresponding process \emph{Bounded Brownian motion}, \cite{Carr17}. Second, the solution $S_t(T)\in[0,1]$ is known in closed form for all $0\leq t\leq T$, $T\in[0,T^*]$, and the probability distribution of $S_t(T)$ is known analytically from \eqref{eq:SolZtTPhi}. 
Third, this model automatically meets the calibration equation, by construction. Indeed, for every $X\sim\mathcal{N}(\mu,\sigma^2)$, it holds
$$\mathbb{E}[\Phi(X)]=\Phi\left(\frac{\mu}{\sqrt{1+\sigma^2}}\right)\;.$$
Hence,
$$\mathbb{E}[S_t]=\mathbb{E}[\Phi(Z_{t,t})]=\Phi\left(\frac{\Phi^{-1}(G(t))e^{\int_0^t\frac{\eta^2(s)}{2}ds}}{\sqrt{e^{\int_0^t\eta^2(s)ds}}}\right)=G(t)\;.$$

The only ``free'' parameter is thus the time-dependent volatility function $\eta$, controlling the randomness of the survival probabilities. As explained above,  sparsity is often considered as an asset in credit derivatives.
By using the properties of $\Phi$, it is easy to show that the diffusion coefficient associated with the $\Phi$-martingale in \eqref{eq:dStT} is $\sigma(t,z)=\eta(t)\phi(\Phi^{-1}(z))$.\medskip 


The associated Az\'ema supermartingale is given by the following It\^{o} process
\begin{equation}
\label{eq:Sexact}
S_t=1+\int_0^te^{\int_0^s\frac{\eta^2(u)}{2}du}\frac{\phi(\Phi^{-1}(S_s))}{\phi(\Phi^{-1}(G(s))}dG(s) + \int_0^t\eta(s)\phi(\Phi^{-1}(S_s))dB_s\;.
\end{equation}
Differentiating \eqref{eq:Sexact} shows that it is a supermartingale satisfying \eqref{eq:azema} with
\begin{equation}
\label{eq:cmcoef}
\lambda_tS_t = 
e^{\int_0^t\frac{\eta^2(u)}{2}du}\frac{\phi(\Phi^{-1}(S_t))}{\phi(\Phi^{-1}(G(t))}h(t)G(t)\quad \text{and} \quad \sigma_t = \eta(t)\phi(\Phi^{-1}(S_t))\;,
\end{equation}

\subsubsection{Default time definition in  martingale models}

The conic martingale setup provides an appealing way to model future survival probability curves with the correct range and therefore, is an interesting alternative to \cite{Ces09}. However, the standard approach in default modeling is to first define $\tau$, e.g. as a first-passage time, and then compute the survival probabilities of interest, as in \eqref{eq:defTauCox} for Cox. At this stage however, it is not clear how one can construct a default time $\tau$ associated with given dynamics for the Az\'ema supermartinagle in the case of conic martingale models. This is not a secondary question as the explicit construction scheme for $\tau$ may help to deal with potential arbitrages issues in the enlargement of filtration setup.\medskip

A natural question to ask is whether one could still use the intensity process $\lambda$ to define $\tau$ as in \eqref{eq:defTauCox} when the martingale part of the Doob-Meyer decomposition of the Az\'ema supermartingale in \eqref{eq:azema} does not vanish. After all, $\Lambda$ is still increasing. It turns out that, generally speaking, this is not correct: defining $\tau$ as in \eqref{eq:defTauCox} leads to a random time which distribution is not compatible with the survival process $S$, as we now show.

\begin{lemma}
\label{lem:consdefault}
Consider a model whose Az\'ema supermartingale $S$ takes the Doob-Meyer decomposition \eqref{eq:azema}.  
and let us note $\tau$ the default time associated with this model. Now, define $\tilde{\tau}$ as in \eqref{eq:defTauCox} where $\lambda$ is the intensity process in \eqref{eq:azema}. Then, $\tau\not\sim \tilde{\tau}$, in general.
\end{lemma}
\begin{proof}
It is clear that $\tilde{\tau}$ is the default time in a Cox setup which survival process solves $d\tilde{S}_t=-\lambda_t\tilde{S}_tdt$ with $\tilde{S}_0=1$, i.e. $\tilde{S}_t=\exp\lbrace-\int_0^t\lambda_s ds\rbrace$. Because the same intensity process $\lambda$ enters both $S$ and $\tilde{S}$, the solution to \eqref{eq:azema} can be written as the multiplicative form $S_t=\tilde{S}_t\tilde{M}_t$ where $\tilde{M}_t:=\exp\lbrace-\int_0^t\frac{\sigma_s^2}{2S^2_s} ds+\int_0^t \frac{\sigma_s}{S_s} dB_s\rbrace$ is a martingale with unit expectation. Indeed, $S_0=\tilde{S}_0M_0=1$ and from It\^{o}'s product rule,
$$dS_t=(-\lambda_t\tilde{S}_tdt)\tilde{M}_t+\tilde{S}_t\left(\frac{\tilde{M}_t\sigma_t}{S_t}dB_t\right)=-\lambda_tS_tdt+\sigma_tdB_t\;.$$
 From the Tower law, we have $\Q(\tilde{\tau}>t)=\E[\tilde{S_t}]$ and $\Q(\tau>t)=\E[S_t]$, where $$\E[S_t]=\E\left[\tilde{S}_t\tilde{M}_t\right]=\mathbb{C}ov\left(\tilde{S_t},\tilde{M}_t\right)+\E\left[\tilde{S}_t\right]\;.$$
 Clearly, $\sigma$ depends on $S$ (in a non-linear way as the coefficient of $dB$ must  vanish when $S\downarrow 0$ or $S\uparrow 1$), hence on $\lambda$. Therefore, $\tilde{S}$ and $\tilde{M}$ both depend on $\lambda$, and there is no reason for their covariance to vanish, in general. 
\end{proof}

\begin{remark}In the limit where $\eta\equiv 0$, so is $\sigma$, each process $S_\cdot(T)$ becomes a trivial martingale (i.e. a constant equal to $G(T)$), and the model becomes deterministic, $S_t=G(t)$. Hence, the survival process solves
$$dS_t=-\lambda(t)S_tdt\;,$$
where the deterministic intensity function satisfies $\lambda(t)=h(t)$. This becomes similar to a Cox model with a deterministic intensity $\lambda$ given by the hazard rate function $h$ associated with $G$. One can then define the default time as in \eqref{eq:defTauCox} with $\lambda_t\leftarrow h(t)$. Similarly, when the intensity process  is deterministic, so is $\tilde{S}$ and $\mathbb{C}ov\left(\tilde{S_t},\tilde{M}_t\right)=0$. Therefore, $\Q(\tau>t)=\E[S_t]=\E[\tilde{S}_t]=\Q(\tilde{\tau}>t)$ and both $\tau,\tilde{\tau}$ have the same survival function given by $G$, showing that in the case of a deterministic intensity, one can define $\tau$ as in \eqref{eq:defTauCox}. But these two examples are special cases. In general, we cannot define $\tau$ using \eqref{eq:defTauCox} in models whose Az\'ema supermartingale  is non-decreasing. 
\end{remark}

\section{Immersion and arbitrages}
\label{sec:arbitrage}

In this paper, we consider several information flows, characterized by the knowledge (or not) of the default indicator. This does not trigger any problem in Cox processes, where an explicit construction scheme is available for $\tau$. In this case indeed, the various filtrations (namely, $\mathbb{F}$, $\mathbb{D}$ and $\mathbb{G}=\mathbb{F}\vee\mathbb{D}$) are well identified. For instance, it is therefore relatively easy to check that the knowledge of the default indicator does not provide a superior information to $\mathbb{F}$ when it comes to pricing default-free assets. This is however much more difficult to verify when there is no explicit definition for $\tau$. In this case indeed, $\mathbb{D}$ and hence $\mathbb{G}$ are not explicitly identified. We refer to \cite{Aksa14} for more details and explicit examples of classical arbitrages using the knowledge of $\tau$. 
To avoid these issues, a basic reduced-form approach under the standard Cox model has already been proposed in order to deal with counterparty risk modeling \cite{Duffie99, Brigo13,Crep2015}. Our purpose is to investigate how to use a non-Cox setup without facing arbitrage opportunities by  using a conic martingale model.

\subsection{Models without martingale part (Cox models)}
As recalled above, default models that belong to the class of Cox models are known to be arbitrage-free. Indeed, it can be shown that there is no arbitrage opportunity provided that every $\mathbb{F}$-martingale remains a $\mathbb{G}$-martingale (see \cite{Jeanb08}). This condition, first introduced under the name of $\mathcal{H}$ hypothesis in \cite{Brem78}, is commonly referred to as the \emph{immersion property}. Cox models provide a very convenient modeling environment in this respect, as they are proven to always satisfy the immersion property. Indeed, it is known (see e.g. \cite[Remark 3.2.1. (iii)]{Biel11}) that the immersion property is equivalent to

\begin{equation}
S_t=S_t(t)=\Q(\tau> t|\mathcal{F}_t)=\Q(\tau> t|\mathcal{F}_{T^*})=S_{T^*}(t)\;.\label{eq:HCox}
\end{equation}

Cox models satisfy the above condition (hence the immersion property):

$$S_{T^*}(t)=\Q(\tau>t|\mathcal{F}_{T^*})=\Q(\Lambda_t\leq \mathcal{E}|\mathcal{F}_{T^*})=\Q(\Lambda_t\leq \mathcal{E}|\mathcal{F}_t)=e^{-\Lambda_t}=S_t$$
for every $t\in[0,T^*]$, where we have used that $\mathcal{E}$ is independent from $\mathcal{F}_{T^*}$ and $\Lambda$ is $\mathbb{F}$-adapted.

\subsection{Models with martingale part (Conic martingale models)}

It is easy to see that the condition \eqref{eq:HCox}
implies that the martingale part in the Doob-Meyer decomposition of $S$ must vanish. Indeed, $S_t(T)$ is decreasing in $T$ for all $t$. In particular, $S_{T^*}(T)$ is decreasing in $T$, too from \eqref{eq:HCox}, so must be $S$. As a consequence, immersion cannot hold if $S$ features a non-trivial martingale part. 
Therefore, existence of potential arbitrage opportunities in such models require more attention compared to Cox models. Interestingly, it has been shown in \cite{CrepSong14} that a suitable redued-form can be applied beyond the immersion setup under some conditions satisfied by the dynamic Gaussian copula (DGC) credit model. A total valuation adjustment (TVA) price process of a general defaultable security based on the dynamic Gaussian copula model have been proposed. In this section, we show that 
the $\Phi$-martingale default model rules out arbitrage opportunities in the sense that it is a particular case of a DGC when an additional condition on the diffusion parameter $\eta$ is satisfied. To do this, we first recall how to define a corresponding default time to the $\Phi$-martingale model. \medskip

Lemma \ref{lem:consdefault}  calls for a procedure to construct $\tau$ in the martingale setup. We show below that this is easy for the $\Phi$-martingale approach, as it fits in the class of Dynamized Gaussian copula models, for which a closed-form expression of the default time has been given; see \cite{crepey2013informationally}.

\begin{definition}[Dynamized Gaussian Copula model]\label{def:DGC}
The dynamized Gaussian copula (DGC) model is a default model where the default time is defined as
\begin{equation}
\label{eq:tauDGC}
\tau=\ell^{-1}\left(\int_0^\infty f(s)dB_s\right)
\end{equation}
where $B$ is an $\mathbb{F}$-adapted Brownian motion, $f$ is a square integrable function with unit $L^2$-norm and  $\ell:\mathbb{R}_+\to \mathbb{R}$ is a differentiable increasing function from  satisfying $\lim_{u\rightarrow 0}\ell(u)=-\infty$ and $\lim_{u\rightarrow \infty}\ell(u)=+\infty$.
\end{definition}
\begin{proposition}
\label{prop:equiv}
The default time associated with the conic martingale model \eqref{eq:cm} 
can be defined as \eqref{eq:tauDGC} if and only if $F=\Phi$ and $\int_0^\infty\eta^2(u)du = +\infty$. 
\end{proposition}

\begin{proof} Let us start by computing $S_t(T)$ in the DGC model.
Suppose that $\tau$ is given by \eqref{eq:tauDGC} and define $\varsigma^2(t):=\int_t^\infty f^2(s)ds$ and $m_t:=\int_0^t f(s)dB_s$. Then,
\begin{equation}
\{\tau>t\} = \left\lbrace \int_t^\infty f(s)dB_s>\ell(t)-m_t\right\rbrace\;.\label{eq:StTDGC}
\end{equation}
The It\^{o} integral is distributed as a zero-mean normal variable with variance $\varsigma(t)$, so that the $\mathbb{F}$-conditional survival process of $\tau$ collapses to 
\begin{equation}
\label{eq:Sphi}
S_t(T) = \mathbb{Q}(\tau>T|\mathcal{F}_t) = \Phi\left(\frac{m_t-\ell(T)}{\varsigma(t)} \right)\;.
\end{equation}
Let us now set $F=\Phi$ and show that this is the form taken by the $\mathcal{F}_t$-conditional survival probability of the event $\tau>T$ in the $\Phi$-martingale model provided that
\begin{equation}
\label{eq:lf}
\ell(u) = - \Phi^{-1}(G(u))\quad \text{and}\quad f(s) = \eta(s)e^{-\int_0^s\frac{\eta(u)^2(u)}{2}du}.
\end{equation}

Using these notations, we can write the solution in \eqref{eq:SolZtTPhi} 
as 
\begin{equation}
Z_{t,T} 
= \frac{Z_{0,T} + \int_0^t\eta(s)e^{-\int_0^s\frac{\eta^2(u)}{2}du}dB_s}{e^{-\int_0^t\frac{\eta^2(s)}{2}ds}}\\
=\frac{\Phi^{-1}(G(T))+m_t}{\varsigma(t)}\;.
\end{equation}
 Indeed, notice that $f$ in \eqref{eq:lf} is of unit $L^2$-norm (as in Definition \ref{def:DGC}) if 
 \begin{equation*}
     e^{-\int_0^\infty\eta^2(u)du} = 0
 \end{equation*}
 or equivalently
 \begin{equation*}
     \int_0^\infty\eta^2(u)du = +\infty \;,
 \end{equation*}
 so that
$$\varsigma^2(t)=\int_t^\infty f^2(s)ds=\int_t^\infty \eta^2(s)e^{-\int_0^s\eta^2(u)du}=e^{-\int_0^t\eta^2(s)ds}.$$
From \eqref{eq:phimart}, 
\begin{equation}
\label{eq:SF}
S_t(T)= \Phi \left(\frac{\Phi^{-1}(G(T))+m_t}{\varsigma(t)}\right)
\end{equation}
 which agrees with \eqref{eq:Sphi}. Note that $f$ defined in \eqref{eq:lf} meets the assumptions given in Definition \ref{def:DGC}.\medskip

Let us now show that $S_t(T)$ associated with the conic martingale default model with $F\neq \Phi$ cannot be written as \eqref{eq:Sphi}. It is easy to show that if $\psi$ is regular enough for \eqref{eq:diffZ} to admit a unique strong solution,
$$S_t(T)=F\left(\frac{m_t-l_t(T)}{\varsigma_t}\right)\;,$$
where
\begin{equation}
\label{eq:ltT}
l_t(T)=-F^{-1}(G(T))+\int_{0}^t\left(\psi(Z_{s,T})-Z_{s,T}\right)\frac{\eta^2(s)}{2}e^{-\int_0^s\frac{\eta^2(u)}{2}du}ds.
\end{equation}
This agrees with \eqref{eq:Sphi} if and only if  $F=\Phi$, leading to $\psi(x)=x$ and $l_t(T)=-F^{-1}(G(T))=l(T)$.
\end{proof}

The next corollary shows that the $\Phi$-martingale model is an arbitrage free default model in the of \cite{CrepSong14} which in addition allows for automatic calibration to CDS market quotes insured by a specific function $\ell$ in \eqref{eq:lf}. 
\begin{corollary}
\label{cor:nonarbitrage}
 The $\Phi$-martingale model is a case of ``non immersion" arbitrage free default model if $\int_0^\infty\eta^2(u)du = +\infty$.
\end{corollary}
\begin{proof}
The result follows with a direct application of Proposition \ref{prop:equiv} and Theorem 6.2 in \cite{CrepSong14}
\end{proof}
Notice that on the top of being arbitrage-free, the $\Phi$-martingale default model features some attractive additional properties in a practical perspective. First the distribution of $S_t(T)$ is known in close form. Second, since in our case, we first specify the dynamics of $Z_{t,T}$, this intuitively implies the choices of $f$ and $\ell$ in \eqref{eq:lf}, where, in particular, $\ell$ is given by the calibration constraint, $\ell(u)=Z_{0,u}:=\Phi^{-1}(S_{0,u})=\Phi^{-1}(G(u))$. Third, another contribution comes from the fact that one has an exact scheme for $S$ and then for $\lambda S$ via \eqref{eq:cmcoef} easing the numerical computation of respectively \eqref{eq:GenCVA} and \eqref{eq:cdspricebis} since $S_t(T)$ can be rewritten as a function of $S_t$ (see Appendix \ref{sec:exactsche} for more details).

\section{Numerical experiments}\label{sec:numerics}
In this section, we provide numerical examples by considering the $\Phi$-martingale default model in \eqref{def:phim} with constant diffusion coefficient $\eta(t)=\eta$ (so that $\int_0^\infty\eta^2(u)du = +\infty$). We choose as benchmark the JCIR++ model devised in section \ref{sec:jcir} which is a very standard approach when it comes to deal with high credit spread \cite{Brigo10}. 


The performances of the $\Phi$-martingale model will be compared to respectively the PS-JCIR (i.e. the JCIR with positive shift constraint) and the TC-CIR (i.e. CIR time-changed in a way such that a perfect fit is achieved) models using real market data when the default counterparty is Ford. We then consider two different applications in credit risk namely the pricing of credit value adjustment (CVA) in the presence of wrong-way risk (WWR) effects and credit default swap options (CDSO) by considering Ford as reference entity. The considered Ford's CDS spreads are presented on the table below. 

\begin{table}[H]
    \centering
    \begin{tabular}{lccccc}
    \hline
     Maturity (years)    & 1 & 3 & 5 & 7 & 10\\
     \hline
     Spread (bps)    & 18.3 & 136.6 & 191.9 & 267.6 & 280.6\\
     \hline
    \end{tabular}
    \caption{CDS spread term structure of Ford Inc. on November 12, 2018. Source: Bloomberg.}
    \label{tab:spread}
\end{table}
As we can see with Table \ref{tab:spread}, the counterparty's term structure is not fully known at each point in time but only provides data at some maturities. In this context, we need further assumptions in order to construct the market curve $G$ associated to the default time $\tau$ of the reference entity. To do so, we assume piecewise constant hazard rates bootstrapped form the CDS spread associated to each maturity of Table \ref{tab:spread}. This is a common market practice procedure known as the JP Morgan model \cite{IsdaCDS}.
Once the market curve $G$ is fully determined, the next step is to calibrate the models parameters to the obtained market curve.  While the $\Phi$-martingale model provide an automatic calibration, the PS-JCIR and the TC-CIR models's parameters need to calibrated to the market curve. This is done using an optimization procedure searching the models parameters that minimize the discrepancies between model and market risk-neutral survival probability curves. Notice that for the spacial case of the PS-JCIR, the optimization problem includes an additional constraint (i.e. $\varphi\geq 0$) ensuring non-negativity of the intensity process governed by the JCIR++ model. The parameters obtained after calibration are 
$$ \kappa=0.0624,\quad \beta=0.2975, \quad \delta=0.3343, \quad x_0=0.0000$$
for the TC-CIR model, and
   $$ \kappa=0.4382,\quad \beta=0.0086, \quad \delta=0.0396, \quad x_0=0.1051$$
for the PS-JCIR model with jumps parameters
\begin{equation}
\label{eq:jumparam}
\omega=9.5619\cdot10^{-10}, \quad \alpha=3.1508\cdot10^{-10}.
\end{equation}
Observe that the jumps parameters of the PS-JCIR model in \eqref{eq:jumparam} are very close to zero.
\footnote{This most likely results from the positivity constraint: large jumps can help boosting the volatility but because they can only go upwards, they need to be compensated by a negative shift for the model curve to remain in line with the market curve. Because of the positivity constraint, the shift can only compensate a jump process with little activity.} Numerical examples about the PS-CIR model showing its limitations to reproduce high WWR effects and CDSO implied volatilities are provided in \cite{Mbaye2019}. In what follows, we will focus on the numerical comparison of the three considered models ($\Phi$-martingale, PS-JCIR and TC-CIR) in term of their ability to feature both WWR and CDSO implied volatilities.

\subsection{Application to wrong-way risk CVA}
Before the 2008 global crisis, the large financial institutions were considered as too big to fail and assumed to be free of default risk. But after the consequences of the crisis resulting to the collapse of Lehman Brothers, counterparty credit risk started to be considered. The associated risk can be priced using CVA which corresponds to the counterparty risk correction to the standard contract's value \cite{Grego10, Biel12}. Hence CVA is a price and not a risk  measure which therefore can be computed using the risk-neutral pricing machinery. As shown in section \ref{sec:arbitrage}, all the considered default models ($\Phi$-martingale, PS-JCIR and TC-CIR) are free of arbitrage opportunities and a general risk-neutral valuation formula of the time-$t$ CVA on $[0,\tau\wedge T]$ (neglecting margin effects) is given according to Corollary \ref{cor:nonarbitrage}  by
\begin{equation}\label{eq:cva}
    {\rm CVA}_t = \beta_t\mathbb{E}\left[(1-R)\frac{V^+_\tau}{\beta_\tau}\mathds{1}_{\{\tau<T\}}\bigg|\mathcal{G}_t\right]
\end{equation}
for a $\mathcal{G}_\tau$-measurable exposure $V$, recovery rate $R$ and bank account $\beta_t=e^{\int_0^tr_sds}$ where $r$ is the risk-free rate that we assume to be driven by the following Vasicek dynamics:
$$dr_t=\gamma(\theta-r_t)dt+\sigma dW_t,\quad r_0\in\mathbb{R} $$
where $\gamma$, $\theta$, $\sigma$ are positive constants and $W$ is an $\mathbb{F}$-Brownian motion.
The time-$t$ CVA formula \eqref{eq:cva} can be expressed in term of $\mathcal{F}_t$ conditional expectations with an integral involving the Az\'ema supermartingale $S_t$. Using this shortcut, the time-0 CVA simly yields
\begin{equation}
    {\rm CVA}=-\mathbb{E}\left[(1-R)\int_0^T\frac{V^+_u}{\beta_u}dS_u\right]
\end{equation}
which is equivalent to the general CVA formula
\begin{equation}
\label{eq:GenCVA}
    {\rm CVA}=\mathbb{E}\left[(1-R)\int_0^T\frac{V^+_u}{\beta_u}\lambda_uS_udu\right]
\end{equation}
where $\lambda_tS_t$ is given by \eqref{eq:cmcoef} for the $\Phi$-martingale model, $\lambda^\varphi_te^{-\int_0^t\lambda^\varphi_udu}$ for the PS-JCIR model and $\lambda^\theta_te^{-\int_0^t\lambda^\theta_udu}$ for the TC-CIR model.

In this example, we consider one of the most traded OTC derivative, an Interest Rate Swaps (IRS). 
An IRS contract implies two counterparties: B (most often a bank), the paper, exchanges a fixed rate $K$  for a floating rate $F$ with a defaultable counterparty C (here, Ford), called the receiver, at say quarterly payments dates $\mathcal{T}_{a+1},\ldots,\mathcal{T}_b$. The contract starts at $\mathcal{T}_a$ and ends at $T=\mathcal{T}_b$. However, if C defaults before the maturity of the contract (at the default time $\tau$), B looses part of the contract and will only receive the recovered part of the exposure . CVA is the expected losses (the non-recovered part) do to the default of C.\\
The discounted payoff at time t, seen from B, of an IRS with a unit notional can be expressed as: 
\begin{itemize}
    \item if $t\leq \mathcal{T}_a$,
$$V_t = \sum_{i=a+1}^b\Delta_iP_t(\mathcal{T}_i)(F_t(\mathcal{T}_{i-1},\mathcal{T}_i)-K)$$
\item if $\mathcal{T}_{j-1}<t\leq\mathcal{T}_{j}$,
$$V_t =(F_{\mathcal{T}_{j-1}}(\mathcal{T}_{j-1},\mathcal{T}_j)-K)\Delta_jP_t(\mathcal{T}_j) + \sum_{i=j+1}^b\Delta_iP_t(\mathcal{T}_i)(F_t(\mathcal{T}_{i-1},\mathcal{T}_i)-K)$$
\end{itemize}
where
$$F_t(\mathcal{T}_{i-1},\mathcal{T}_i) = \frac{P_t(\mathcal{T}_{i-1})-P_t(\mathcal{T}_i)}{P_t(\mathcal{T}_i)\Delta_i}$$
is the forward rate prevailing at $t$ between $\mathcal{T}_{i-1}$ and $\mathcal{T}_i$ and $P_t(T)=\mathbb{E}[\beta_t/\beta_T|\mathcal{F}_t]$ given in \eqref{eq:PtT} is the time-$t$ risk free zero-coupon bond price with maturity $T$.\medskip

In general, the exposure process $V$ cannot be considered as independent from $\tau$. 
This dependency is known as wrong-way risk (WWR) and is handled here by correlating the Brownian motions $B$ and $W$ entering the dynamics of the exposure and the process governing the default, i.e. $dB_t dW_t=\rho dt$. In the independent case (no-WWR or $\rho=0$), CVA only depends on the recovery rate, the survival probability curve extracted from the market $G$ and the discounted expected positive exposure (EPE) $\mathbb{E}\left[\frac{V^+_t}{\beta_t}\right]$. Under this assumption, the CVA formula is simply given by:
\begin{equation}\label{eq:CVAperp}
    {\rm CVA}^\perp = -(1-R)\int_0^T\mathbb{E}\left[\frac{V^+_u}{\beta_u}\right]dG(u)\;.
\end{equation}
Figure \ref{fig:EPE} depicts the typical EPE profile of a 5-year IRS struck at the prevailing swap rate with respect to time. 
\begin{figure}[H]
\centering
\includegraphics[width=0.75\textwidth]{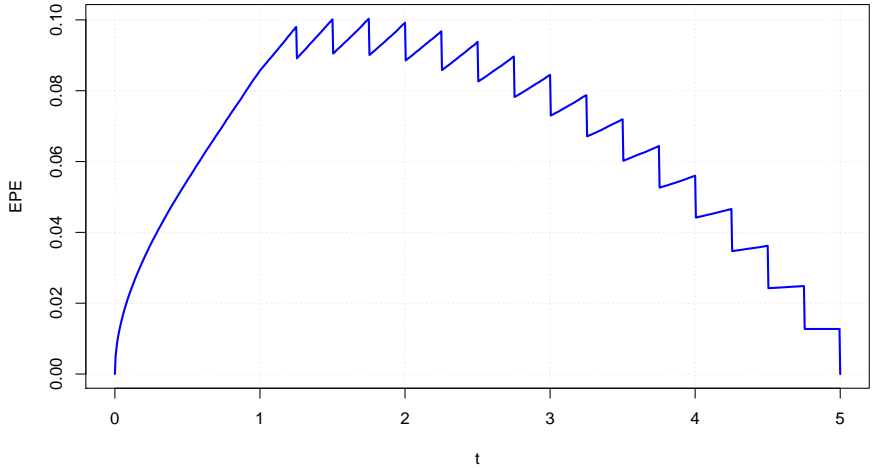}
\caption{Discounted expected positive exposure (EPE) computed by Monte Carlo simulation using $5\cdot10^5$ paths and time step of 0.1\%. IRS with Vasicek parameters $\gamma=0.4$, $\theta=0.026$, $\sigma=0.14$ and $r_0=0.0165$ fitted to a flat interest rate market curve with constant yield of 3\%. We used $\mathcal{T}_a=1$, $\mathcal{T}_b=5$, fixed rate $K=F_0$ and quarterly payment dates $\Delta_i=0.25$.}\label{fig:EPE}
\end{figure}

Under no-WWR, the default model plays no role. The CVA formula \eqref{eq:CVAperp} can be decomposed in the recovery rate, the exposure and the default components. In particular, the only thing that matters regarding $\tau$ is the market curve $G$. This means that whatever the default model considered, the CVA will remain unchanged provided that the model is calibrated to the market curve $G$. The picture completely changes under WWR as the EPE needs to be replaced by a \emph{conditional} EPE. This triggers a dependency to both the dynamics of the default model and the dependence parameter $\rho$. We refer to~\cite{Brigo17} for more details about the management of WWR, including some analytical approximations of conditional EPEs.\medskip

In Figure \ref{fig:FigCVA}, we plot the CVA as a function of $\rho$ for the three considered models: PS-JCIR (solid green), TC-CIR (dotted blue) and $\Phi$-martingale (dashed magenta with $\eta\in\{0.10, 0.15, 0.20, 0.50, 0.75\}$), calibrated to the same survival probability curve $G$ given by Ford's CDS term structure. Thus due to the calibration constraint, all models agree with the special case of no-WWR ($\rho=0$): the independent CVA lined up in cyan. Further, we observe that CVA generally increases with the correlation parameter $\rho$ for all the considered models meaning that all the models are able to reproduce the WWR effect. However, if we evaluate the models performances in term of their capability to feature high WWR, one can notice that the PS-JCIR is less competitive. Because of the positivity constraint, the model features the lowest WWR impact which is comparable to the one featured by the $\Phi$-martingale with $\eta=10\%$. In particular, when increasing the volatility parameter $\eta$, the WWR produced by the $\Phi$-martingale model increases as well and is comparable to the one generated by the TC-CIR model when $\eta=75\%$. Hence, the $\Phi$-martingale is a simple and tractable model allowing for automatic calibration to CDS quotes, and able to generate a wide range of WWR effects  by playing with the parameter, $\eta$. Although it is possible to extend the model to make $\eta$ a time-dependent, the sparsity of the model is a nice feature for illiquid products, like CDS options.

\begin{figure}[H]
\centering
\includegraphics[width=0.75\columnwidth]{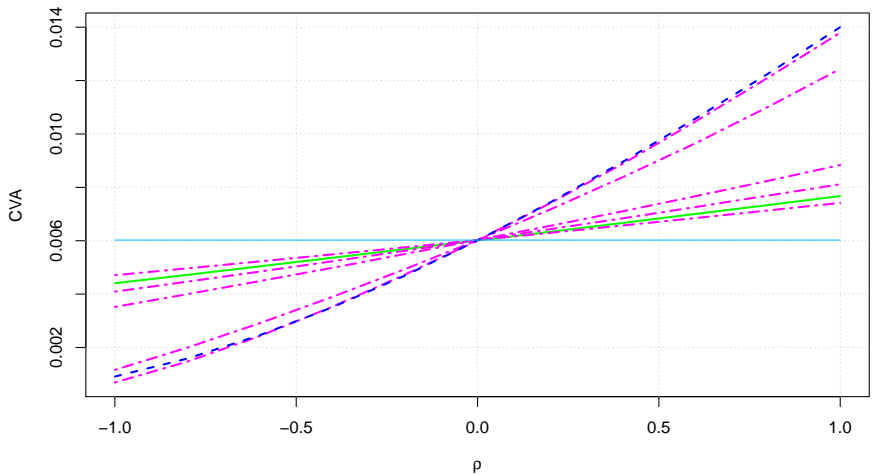}
\caption{CVA figures as a function of the correlation $\rho$ for PS-JCIR (solid green), TC-CIR (dotted blue), $\Phi$-martingale (dotted magenta, $\eta\in\{0.10,0.15,0.20,0.50,0.75\}$) and the CVA with zero-correlation (solid cyan) using Monte Carlo method with $10^6$ paths, time step $0.01$. Profiles: 5Y IRS exposure detailed in \ref{fig:EPE}.}\label{fig:FigCVA}
\end{figure}

\subsection{Volatility surface of CDS option}
\label{sec:}
We now proceed to the assessment of the volatility of CDS spreads in term of CDS option implied volatilities. A CDS (call) option with maturity $T_a$ on a single name gives the invertor the right to enter at $T_a$ a \emph{payer} CDS on a single name with contractual spread $k$ and terminantion time $T_b>T_a$. The corresponding no-arbitrage price (as explained in \cite{Brigo10}) is given at time $t=0$ by:
\begin{equation}
\label{eq:cdsoprice}
    PSO(a,b,k)=\mathbb{E}[\beta_{T_a}^{-1}(CDS_{T_a}(a,b,k))^+]
\end{equation}
in which $CDS_t(a,b,k)$ stands for the time-$t$ pre-default value of the underlying CDS starting at time $T_a$ with maturity $T_b$ given by
\begin{equation}
\label{eq:cdsprice}
    CDS_t(a,b,k)=\mathds{1}_{\{\tau>t\}}\left(-(1-R)\int_{T_a}^{T_b}P_t(u)\partial_u \overline{Q}_t(u)du-k\,C_t(a,b)\right)
\end{equation}
where we assume independence between the risk-free rate and the default intensity processes,
\begin{equation}
    \overline{Q}_t(T):=\mathds{1}_{\{\tau>t\}}Q_t(T)=\frac{S_t(T)}{S_t}
\end{equation}
and
\begin{equation}
    C_t(a,b):=\sum_{i=a+1}^b\alpha_iP_t(T_i)\overline{Q}_t(T_i)-\int_{T_{i-1}}^{T_i}\frac{u-T_{i-1}}{T_i-T_{i-1}}\alpha_iP_t(u)\partial_u \overline{Q}_t(u)du
\end{equation}
is the time-$t$ value of the CDS premia paid during the life of the contract when the spread is 1 known as the risky duration.

Setting expression \eqref{eq:cdsprice} to zero and solving in $k$ yields the expression for the par CDS spread
\begin{equation}
    s_t(a,b):= \frac{-(1-R)\int_{T_a}^{T_b}P_t(u)\partial_u \overline{Q}_t(u)du}{C_t(a,b)}.
\end{equation}
Finally, plugging \eqref{eq:cdsprice} into \eqref{eq:cdsoprice}, the time-0 CDS option price can be simply rewritten as
\begin{equation}
\label{eq:cdspricebis}
    PSO(a,b,k)=\beta_{T_a}^{-1}\mathbb{E}\left[S_{T_a}\left((1-R) - \sum_{i=a+1}^b\int_{T_{i-1}}^{T_i}g_i(u)P_{T_a}(u)\overline{Q}_{T_a}(u)du \right)^+\right]
\end{equation}
where $g_i(u):=(1-R)(r(u)+\delta_{T_b}(u))+k\frac{\alpha_i}{T_i-T_{i-1}}(1-(u-T_{i-1})r(u))$, with $\delta_s(\cdot)$ is the Dirac delta function centred at $s$.\\
Clearly, formula \eqref{eq:cdspricebis} can be implemented using the three considered models ($\Phi$-martingale, PS-JCIR and TC-CIR) by considering the corresponding conditional survival process of each model.

A CDS option is typically quoted on the market in term of its Black implied volatility $\bar{\sigma}$ which is based on the assumption that the credit spread follows a geometric Brownian motion. The Black formula for payer swaptions at time 0 with maturity $T_a$ is
\beq
PSO^{Black}(a,b,k,\bar{\sigma}) = C_0(a,b)\left[s_0(a,b)\Phi(d_1)-k\Phi(d_2)\right]\nonumber
\eeq
where
\beq
d_1=\frac{\ln\frac{s_0(a,b)}{k}+\frac{1}{2}\bar{\sigma}^2T_a}{\bar{\sigma}\sqrt{T_a}},\qquad d_2=d_1-\bar{\sigma}\sqrt{T_a}\nonumber
\eeq
and $\Phi$ is the distribution function of a standard Normal random variable.\\
Hence, the CDS option implied volatility $\bar{\sigma}$ can be found by solving the following equation
\begin{equation}
    PSO(a,b,k)=PSO^{Black}(a,b,k,\bar{\sigma}).
\end{equation}
 Given the forward spread and risky annuities, we can compute the implied volatilities for payers written on the same underlying CDS for the three models ($\Phi$-martingale, PS-JCIR and TC-CIR) using at-the-money payer (Table \ref{tab:CDSO}) or with different strikes (Table \ref{tab:strikeImpVol}). We have assumed zero interest rates in the numerical applications since we are focusing on the impact of the default model. As expected, Table \ref{tab:CDSO} shows that the PS-JCIR model generates small implied volatilities compared the other models. In addition, it is not difficult to notice that the TC-CIR model features much more implied volatilities compared to the PS-JCIR but the level of volatility remains relatively small.  It could be possible to increase level of volatility implied by the TC-CIR by playing with the model's parameters but Feller's constraint is required and puts limits to the volatility magnitude that can be achieved. In contrast, the $\Phi$-martingale model gives the freedom to increase the volatility level without facing any constraint and this in turn allows to increases the CDS option implied volatility without bounds. This is very important when dealing with a counterparty of poor credit quality characterized by a high credit risk as it was the case of many firms in the recent global crisis, a property that most existing models fail to capture.  
 \begin{table}[H]
\centering
\resizebox{10cm}{!}{
\begin{tabular}{|c|c|c|c|c|c|c|c|}
   \hline
    $T_a$&$T_b$& PS-JCIR (\%) & TC-CIR (\%)& \multicolumn{4}{c|}{$\Phi$-martingale (\%)}   \\
   \hline
   & & & & $\eta=10$\% & $\eta=15$\% & $\eta=20$\% & $\eta=50$\%   \\
   \hline
   1&3 &19.55 &44.00& 20.24 & 30.11 & 39.91 & 99.46\\
   \hline
   1&5 &11.90 &26.56& 17.31 & 25.64 & 34.00 & 83.86\\
   \hline
  1& 7 &7.33 &16.30& 14.82 & 21.79 & 28.74 & 70.59\\
   \hline
  1 &10 &5.72 &12.27& 13.64 & 20.06 &  26.40 & 64.52\\
   \hline
   3 &5 &7.79 &58.67& 15.06 & 22.40 & 29.84 & 76.67\\
   \hline
   3&7 &4.48 &36.10&13.24 & 19.60 & 26.03 & 66.34 \\
   \hline
   3&10 &3.47 &26.62& 12.36 & 18.32 & 24.36 & 61.59\\
   \hline
  5& 7 &3.11 &43.32& 12.40& 18.39 & 24.55 & 65.34\\
   \hline
   5&10 &2.49 &30.61& 11.58 & 17.24 & 23.02 & 61.29\\
   \hline
   7&10&2.47 &39.40& 10.24 & 15.33 & 20.61& 58.56\\
   \hline
\end{tabular}
}
\caption{Black volatilities for at-the-money ($k=s_0(a,b)$) payer CDS options implied by the PS-JCIR, TC-CIR and the $\Phi$-martingale models using Monte Carlo simulation ($2.10^6$ paths with time step 0.01) for various volatility parameter $\eta$.} 
\label{tab:CDSO}
\end{table}
Table \ref{tab:strikeImpVol} shows how the CDS option price evolves when increasing the strike $k$. We observe first that the payer CDS option decreases when the strike prices increases which is evident since a payer with higher strike price is worthless. Alternatively, the $\Phi$-martingale model seems to be more sensitive with respect to the strike and can give different levels of price even with lower or higher strike with the help of the parameter $\eta$. In contrast, the parameters of the other models are fixed due to the calibration constraint (PS-JCIR and TC-CIR) or the positivity constraint (PS-JCIR).

\begin{table}[H]
\centering
\resizebox{7.5cm}{!}{
\begin{tabular}{|c|c|c|c|c|}
   \hline
    $k$ (bps)& PS-JCIR  & TC-CIR & \multicolumn{2}{c|}{$\Phi$-martingale }   \\
   \hline
    & & & $\eta=15$\% & $\eta=20$\%    \\
   \hline
   200 & 148.96 & 152.07 & 176.12 & 198.56\\
   \hline
   220 & 85.64 & 116.54 & 126.99 & 153.89\\
  \hline
  240 & 39.55 & 92.34 & 88.06 & 117.27\\
  \hline
  260 & 14.08 & 74.18 & 58.67 & 87.68\\
  \hline
  280 & 3.83 & 60.16 & 37.84 & 64.50\\
  \hline
  300 & 0.79 & 49.26 & 23.52 & 46.95\\  
   \hline
\end{tabular}
}
\caption{European payer (bps) with maturity $T_a=1$ year to enter into a single-name CDS (Ford Inc.) with $T_b=5$ year maturity with different strikes implied by the PS-JCIR, TC-CIR and $\Phi$-martingale models using Monte Carlo method with $10^6$ paths and time step $0.01$. The forward spread is $s_0(a,b)=238$ bps and the volatility parameter of the $\Phi$-martingale model is $\eta=0.15.$}
\label{tab:strikeImpVol} 
\end{table}

  
  

\section{Conclusion}\label{sec:concl}
The most popular  class of credit risk  models is undoubtedly  the set of Cox models with stochastic default intensities governed by positive dynamics such as CIR or JCIR. If the intensity dynamics are simple enough (like time-homoegeneous square-root diffusions), these models allow for closed form solutions for the prices of defaultable bonds or even options, and are arbitrage-free by construction since the immersion property is satisfied. However, this simplicity comes at the price of drawbacks that are manifest in actual credit risk applications. The first one is that, given the few number of parameters at hand, such models are not flexible enough to allow a good fit to the prices prescribed by the market. To circumvent the calibration issues, the deterministic shift extension, leading in particular to CIR++ and JCIR++ models, is a very good alternative but is often problematic in practice as it features negative intensities. Adding a non-negativity constraint is a simple fix, but which often drastically limits the model's volatility. A time-change version seems to help in this respect, as it allows to get a perfect fit while generating  volatilities being somewhat larger.  Second, a modeling framework that satisfies the immersion property is quite a specific configuration, corresponding to a decreasing Az\'ema supermartingale. 

In this paper, we have developed a defaultable term structure model that does not fit in the class of Cox models. The conic martingale approach seems a promising alternative to the standard stochastic default intensity model in this respect. It fills the gap of negative intensities since all survival probabilities are bounded and belong to $[0,1]$ by construction. Moreover, although it goes beyond immersion, 
we have shown that a particular case of conic martingales models, the $\Phi$-martingale approach, 
is a special case of the dynamic Gaussian copula (DGC) introduced by Cr\'epey et al \cite{CrepSong14}, and therefore is arbitrage-free. Furthermore, this allows us to identify an explicit construction scheme for the default time.  Eventually, the model is sparse, and can exhibit a large volatility impact, which is interesting for option and counterparty credit risk pricing.  
These interesting features have been illustrated in two examples: the valuation adjustment under counterparty credit (CVA) risk and the pricing of CDS option. 
For all these reasons, the $\Phi$-martingale model seems to be an appealing trade-off between theory and practice and provides an interesting tool for credit risk practitioners such as CVA traders and risk managers.\medskip

Dealing with arbitrage for a general conic martingales model seems to be less evident. Nevertheless, additional conditions on the score function $\psi$ show that conic martingales belong to the density models of El Karoui et al. \cite{Elk10} which are examples of beyond immersion models satisfying the $\mathcal{H}'$-hypothesis (i.e. the martingales on the default-free filtration are semimartingales in the full filtration). Future research will investigate the special case of the positive density hypothesis, a sufficient condition ensuring the no-arbitrage property in this class of models.

\section{Appendix: Exact scheme of $S$ via $Z$}

\label{sec:exactsche}
In this section, we show that is possible to simulate exactly the survival process $S$ with the $\Phi$-martingale model in contrast with the other models considered in this paper that require a full truncation scheme such as Euler to be simulated. Let's consider \eqref{eq:SolZtT} by setting the diffusion coefficient $\eta(t)=\eta$ to be constant to simplify the exposition. One gets
\begin{equation}
\label{eq:Ztt}
   Z_t:= Z_{t,t}=\Phi^{-1}(G(t))e^{\frac{\eta^2}{2}t} + \eta\int_0^te^{\frac{\eta^2}{2}(t-s)}dB_s\;.
\end{equation}
Using \eqref{eq:Ztt} and considering a time step $\Delta$, we can express the exact distribution of $Z_{t+\Delta}$ in term of $Z_t$:
\begin{eqnarray}
Z_{t+\Delta} &=&e^{\frac{\eta^2}{2}(t+\Delta)}\left((\Phi^{-1}(G(t+\Delta))+\eta\int_0^te^{-\frac{\eta^2}{2}s}dB_s + \eta\int_t^{t+\Delta}e^{-\frac{\eta^2}{2}s}dB_s\right)\nonumber\\
&=& e^{\frac{\eta^2}{2}(t+\Delta)}\left(\Phi^{-1}(G(t+\Delta))-\Phi^{-1}(G(t))+Z_te^{-\frac{\eta^2}{2}t} + \eta\int_t^{t+\Delta}e^{-\frac{\eta^2}{2}s}dB_s\right)\nonumber\\
&\sim&Z_te^{\frac{\eta^2}{2}\Delta} + \left[\Phi^{-1}(G(t+dt))-\Phi^{-1}(G(t)) \right]e^{\frac{\eta^2}{2}(t+\Delta)} + \sqrt{e^{\eta^2 \Delta}-1}\,Y_t\;.\nonumber
\end{eqnarray}
where $Y_t\sim\mathcal{N}(0,1)$ is independent from $Z_t$.\\
From the latter, we can deduce the exact simulation of $S$ by setting $S_t=\Phi(Z_t)$.\\
In addition, the survival process $S_t(T)$ the can also be derived from $S_t$. Indeed, using \eqref{eq:Ztt} again,
\begin{eqnarray}
Z_{t,T} &=&\Phi^{-1}(G(T))e^{\frac{\eta^2}{2}t} + \eta\int_0^te^{\frac{\eta^2}{2}(t-s)}dB_s\nonumber\\
&=& Z_t + \left[\Phi^{-1}(G(T))-\Phi^{-1}(G(t)) \right]e^{\frac{\eta^2}{2}t}
\end{eqnarray}
so that
$$S_t(T) =\Phi\left( \Phi^{-1}(S_t)+ \left[\Phi^{-1}(G(T))-\Phi^{-1}(G(t)) \right]e^{\frac{\eta^2}{2}t}\right)$$
and finally
$$\overline{Q}_t(T)=\frac{\Phi\left(\Phi^{-1}(S_t)+ \left[\Phi^{-1}(G(T))-\Phi^{-1}(G(t)) \right]e^{\frac{\eta^2}{2}t}\right)}{S_t}\;.$$

\bibliography{MyBib}
\bibliographystyle{plain}

\end{document}